\documentclass[pra,twocolumn,showpacs,preprintnumbers,amsmath,amssymb,superscriptaddress]{revtex4}

\usepackage{bm}
\usepackage{graphicx}
\usepackage{amsbsy}
\usepackage{amsmath}
\usepackage{amsfonts}
\usepackage{amsthm}

\begin{document}

\theoremstyle{plain}
\newtheorem{theorem}{Theorem}
\newtheorem{lemma}[theorem]{Lemma}
\newtheorem{corollary}[theorem]{Corollary}
\newtheorem{conjecture}[theorem]{Conjecture}
\newtheorem{proposition}[theorem]{Proposition}

\theoremstyle{definition}
\newtheorem{definition}{Definition}

\theoremstyle{remark}
\newtheorem*{remark}{Remark}
\newtheorem{example}{Example}

\title{Necessary and sufficient conditions for local manipulation of multipartite pure quantum states}   

\author{Gilad Gour}\email{gour@ucalgary.ca}
\affiliation{Institute for Quantum Information Science and 
Department of Mathematics and Statistics,
University of Calgary, 2500 University Drive NW,
Calgary, Alberta, Canada T2N 1N4} 
\author{Nolan R. Wallach}\email{nwallach@ucsd.edu}
\affiliation{Department of Mathematics, University of California/San Diego, 
        La Jolla, California 92093-0112}

\date{\today}

\begin{abstract} 
Suppose several parties jointly possess a pure multipartite state, $|\psi\rangle$. 
Using local operations on their respective systems and classical communication (i.e. LOCC) it may be possible for the parties to transform deterministically $|\psi\rangle$ into another joint state $|\phi\rangle$. In the bipartite case, Nielsen majorization theorem gives the necessary and sufficient conditions for this process of entanglement transformation to be possible.
In the multipartite case, such a deterministic local transformation is possible only if both the states in the same
stochastic LOCC (SLOCC) class. Here we generalize Nielsen majorization theorem to the multipartite case, and find necessary and sufficient conditions for the existence of a local \emph{separable} transformation between two multipartite states in the same SLOCC class. When such a deterministic conversion is not possible, we find an expression for the maximum probability to convert one state to another by local separable operations. In addition, we find necessary and sufficient conditions for the existence of a separable transformation that converts a multipartite pure state into one of a set of possible final states all in the same SLOCC class. 

Our results are expressed in terms of  (1) the stabilizer group of the state representing the SLOCC orbit,
and (2) the associate density matrices (ADMs) of the two multipartite states. The ADMs play a similar role
to that of the reduced density matrices, when considering local transformations that involves pure bipartite states.  
We show in particular that the requirement that one ADM majorize another is a necessary condition but in general far from being also sufficient as it happens in the bipartite case. In most of the results the twirling operation with respect to the stabilizer group (of the representative state in the SLOCC orbit) plays an important role that provides a deep link between entanglement theory and the resource theory of reference frames.
\end{abstract}  


\maketitle


\section{Introduction}
Every restriction on quantum operations defines a resource theory, determining how quantum states that cannot be prepared under the restriction may be manipulated and used to circumvent the restriction. Much of
quantum information theory can be viewed as a theory of interconversions among such resources~\cite{Dev08}. 
When several spatially separated parties sharing a composite quantum system, they are naturally restricted to act locally
on the system. Such a restriction to local operations assisted by classical communication (LOCC) leads to the theory of entanglement, in which multipartite entangled states act as resources for overcoming the LOCC restriction. 
Hence, exotic multipartite entangled states play an important role in a variety of quantum information processing tasks. These include conventional~\cite{Nie00} and measurement-based quantum computation~\cite{Rau01}, quantum error correction schemes~\cite{Nie00}, 
quantum secret sharing~\cite{Hil99}, quantum simulations~\cite{Llo96}, and in principle in the description of every composite system consisting of more than one subsystem~\cite{Hor09,Ple07}.  

The amount of information needed to describe $n$-party quantum 
system grows exponentially with $n$, which makes it very difficult and almost impossible to classify multipartite states
by the degree of their resourcefulness (i.e. by their degree to circumvent the limitation of LOCC). The main reason 
is that the strength of the LOCC restriction increases with $n$, and, in general, the extent to which quantum resources 
(like multipartite entangled states) can be interconverted decreases as the strength of the restriction increases~\cite{GS}. 
It can be shown~\cite{Vid00,GW10} that for $n>3$, the LOCC restriction leads to an uncountable number of inequivalent resources (i.e. multipartite entangled states).

A conventional way to classify multipartite states is to divide the Hilbert space into sets of states that are all
related to each other by invertible stochastic local operations and classical communication (SLOCC)~\cite{Vid00}. 
More precisely, we say that $\psi$ and $\phi$ belong to the same (invertible) SLOCC class if both transformations 
$\psi\rightarrow\phi$ and $\phi\rightarrow\psi$ can be achieved by means of LOCC with non-zero probability. 
Hence, in general, if $\psi$ and $\phi$ belong to two distinct SLOCC classes, they are incomparable and essentially correspond to  two deferent types of resources.  While in three qubits there are only 6 inequivalent SLOCC classes (see~\cite{Vid00}), there are uncountable number
of inequivalent SLOCC classes in four (or more) qubits~\cite{Vid00,GW10}. Roughly speaking, this fact implies that relative to the dimension of the Hilbert space, only a small number of transformations can be achieved by local means 
(since the states involved must belong to the same SLOCC class). Nevertheless, the very few transformations
that can be achieved locally are extremely useful for QIP tasks, as in the examples of transformations involving the cluster states~\cite{Rau01}, graph states~\cite{Hei05}, code states~\cite{Nie00}, and other novel multipartite entangled states~\cite{Hor09}.

A fundamental problem in quantum information is, therefore, to characterize all possible deterministic LOCC transformations among states in the same SLOCC class. That is, if $\psi$ and $\phi$ belong to the same SLOCC class, what are the necessary and sufficient conditions that the transformation $\psi\rightarrow\phi$ can be achieved by LOCC with 100\% success (i.e with probability one)? Despite its importance, this problem was solved completely only for the case of pure bipartite 
states~\cite{Hor09,Ple07}. In fact, even in the case of 3 qubits, very little is known about the characterization of local transformations among states in the same GHZ class~\cite{Spe01}. For higher number of qubits, or in higher dimensions, much less is known and there are only partial results for very specific SLOCC classes of states~\cite{Kin10,Cui10}. 

One of the several difficulties encountered in the effort to solve this problem, is the non-elegant mathematical description of (the operationally motivated) LOCC protocols in composite quantum systems consisting of more than two parties. This difficulty has been avoided in the literature by considering a larger set of transformations, such as the positive partial transpose preserving operations~\cite{Ish05}, or the non-entangling operations in~\cite{Bra08}. However, these sets of transformations are much larger than what can be achieved by LOCC, and therefore include many transformations that can not be achieved locally. A much smaller set, that on one hand is local and include all LOCC transformations while on the other hand is mathematically simple to characterize, is the set of local separable transformations (SEP). 

SEP can be characterized as a set of trace preserving quantum channels such that for each channel in the set, there exists an operator sum representation (or Kraus representation) with the property that each Kraus matrix is a tensor product of $n$ matrices acting on the $n$ subsystems consisting the composite system (see the section~\ref{pre} for the mathematical definition of SEP). 
In~\cite{Ben99} it was shown (somewhat counter intuitively) that SEP is strictly larger than LOCC, while in~\cite{Ghe08} it was shown that for bipartite pure states we have effectively SEP=LOCC; that it, any transformation between two bipartite pure states that can be achieved by SEP can also be achieved by LOCC and vice versa. However, for mixed bipartite states this is no longer true~\cite{Eri09}.

In this paper we characterize the set of all transformations taking one multipartite pure state to another by SEP. 
Clearly, unless the two states belong to the same SLOCC class, such a transformation is impossible. 
For each SLOCC class we show that there is a natural state (known also as a normal form~\cite{Ver03}) that can be chosen to represent the class, and use the stabilizer group of this state to present our results. 
Unlike the stabilizer formalism in quantum error correction schemes~\cite{}, here the stabilizer group is not a subgroup of the Pauli group, but, in general, it is a subgroup of a much larger group, the group of all invertible matrices. 

When it is impossible to convert by SEP one state to another, we find an expression for the maximum probability for the conversion. Our expression is in general hard to calculate and we therefore provide simple lower and upper bounds on the maximum probability of conversion. In addition, we generalize our results
to include necessary and sufficient conditions for the existence of a separable transformation that converts a multipartite pure state into one of a set of possible final states all in the same SLOCC class.

The paper is organized as follows. In section~\ref{pre} we discuss some preliminary concepts and theorems about SLOCC classes, separable operations, and normal forms. In section~\ref{sec:stab} we define the stabilizer group and prove several of its 
properties. We also define the stabilizer twirling operation, and discuss the relation between the stabilizer group and SL-invariant polynomials. In section~\ref{sec:asso} we define the associate density matrix and discuss its relation to the reduced density matrix.
In Section~\ref{sec:main} we present and prove our main results. 
Section~\ref{sec:ex}
is devoted to examples demonstrating how to apply our results to specific cases. Finally, in section~\ref{sec:con} we end with conclusions.

\section{Preliminaries}\label{pre}
The Hilbert space of $n$ qudits is denoted by:
$$
\mathcal{H}_n\equiv \mathbb{C}^{d_1}\otimes\mathbb{C}^{d_2}\otimes\cdots\otimes\mathbb{C}^{d_n}\;.
$$
The SLOCC groups we will consider are the subgroups
\begin{align*}
& G\equiv \text{SL}(d_1,\mathbb{C})\otimes\text{SL}(d_2,\mathbb{C})\otimes\cdots\otimes\text{SL}(d_n,\mathbb{C})\\
& \tilde{G}\equiv \text{GL}(d_1,\mathbb{C})\otimes\text{GL}(d_2,\mathbb{C})\otimes\cdots\otimes\text{GL}(d_n,\mathbb{C})
\end{align*}
in $GL(\mathcal{H}_{n})$. Here $GL(d,\mathbb{C})$ is the group consisting of all $d\times d$ complex invertible matrices
and $SL(d,\mathbb{C})$ is a subset of these matrices with determinant 1.

Viewing the SLOCC classes as orbits over the action of $G$, we define the orbit of a state $\psi\in\mathcal{H}_n$ to consists of unit vectors:
\begin{definition}
The orbit $\mathcal{O}_{\psi}$ of a state $|\psi\rangle\in\mathcal{H}_n$ under the action of $G$ is defined by:
$$
\mathcal{O}_{\psi}\equiv\left\{\frac{g|\psi\rangle}{\|g|\psi\rangle\|}\Big|\;g\in G\right\}\;.
$$
\end{definition}
Note that the orbit above if unchanged if one replace $G$ by $\tilde{G}$ since we ignore global phase.
Note also that any state in the orbit can be chosen to represent the orbit; that is, if $\psi'\in\mathcal{O}_{\psi}$ then $\mathcal{O}_{\psi'}=\mathcal{O}_{\psi}$.
Sometimes in literature the term ``SLOCC class" is referred to a similar set as above, but slightly different, allowing the matrix $g$ in the definition above not to be invertible.

\subsection{Separable Matrices and Separable Operations}

The set of all states that can be prepared by LOCC is called the set of separable states.
Here we define it without the normalization requirement of trace one.
\begin{definition}
A matrix (or operator) $\rho:\mathcal{H}_n\to\mathcal{H}_n$ is said to be a separable matrix or a separable operator
if it can be written as:
$$
\rho=\sum_{k}\sigma_{1}^{(k)}\otimes \sigma_{2}^{(k)}\otimes\cdots\otimes \sigma_{n}^{(k)}
$$
where all $\sigma_{i}^{(k)}:\mathbb{C}^{d_i}\to\mathbb{C}^{d_i}$ are positive semi-definite matrices (operators).
\end{definition}

Later on, we will use the fact that any tensor product of matrices can be completed to the identity matrix by an addition of a separable matrix:
\begin{lemma}\label{flem}
Let $M=A_1\otimes A_{2}\otimes\cdots\otimes A_n$, where $A_k$ are $d_k\times d_k$ positive semi-definite matrices.
Denote by $\lambda_{\max}$ the largest eigenvalue of $M$. Then, for $p\in\mathbb{R}$, the matrix $I-pM$ is separable if and only if 
$p\leq 1/\lambda_{\max}$.
\end{lemma}
\begin{proof}
Clearly $pM\leq I$ so $I-pM$ is positive semi-definite. To see that it is also separable denote $\lambda_j$ to be the maximum eigenvalue of $A_j$, denote $I_j$ to be the identity matrix on subsytem $j$, and w.l.o.g take $p=1/\lambda_{\max}=\left(\lambda_1\lambda_2\cdots\lambda_n\right)^{-1}$. Then, the matrix
$I_j-\frac{1}{\lambda_j}A_j$ is positive semi definite, and therefore $I_j$ is a sum of two positive semi definite matrices $I_j=(I_j-\frac{1}{\lambda_j}A_j)+\frac{1}{\lambda_j}A_j$. Substituting this into $I=I_1\otimes I_2\otimes\cdots I_n$ completes the proof.
\end{proof}

We now give the precise definition of the set of operations that will be used in this paper.

\begin{definition}
Let $\mathcal{B}(\mathcal{H}_n)$ be the set of density matrices acting on $\mathcal{H}_n$, and
let $\Lambda:\mathcal{B}(\mathcal{H}_n)\to\mathcal{B}(\mathcal{H}_n)$ be a completely positive trace-preserving map.
Then, we say that $\Lambda$ is a \emph{separable} super-operator if there exists an operator sum representation of $\Lambda$ such that for any density matrix $\rho\in\mathcal{B}(\mathcal{H}_n)$,
\begin{equation}\label{sep}
\Lambda(\rho)=\sum_{k}M_{k}\rho M_{k}^{\dag}\;,
\end{equation}
and 
$$
M_k=A_{1}^{(k)}\otimes A_{2}^{(k)}\otimes\cdots\otimes A_{n}^{(k)}\;,
$$
where $A_{i}^{(k)}$ are $d_i\times d_i$ complex matrices.
\end{definition}

We denote the set of all separable superoperators by SEP. In the particular case, when considering only pure bipartite states, 
LOCC is effectively the same as SEP, but in general, LOCC is a strict subset of SEP.

\subsection{The Critical Set}

Let $\mathfrak{g}$ be the Lie algebra of $G$
contained in $End(\mathcal{H}_{n})$. We define the set of \emph{critical} states as
$$
Crit(\mathcal{H}_{n})\equiv\{\phi
\in\mathcal{H}_{n}|\left\langle \phi|X|\phi\right\rangle =0,X\in
\mathfrak{g\}}.
$$ 
The GHZ states~\cite{Vid00}, cluster states~\cite{Rau01}, graph states~\cite{Hei05}, and code states~\cite{Nie00} (such as the 5-qubit codes state, the Steane 7-qubit code state,
and the Shore 9-qubit code states) are all critical. In the bipartite case, with $\mathcal{H}_2=\mathcal{C}^{d}\otimes\mathcal{C}^{d}$,
the maximally entangled state is the only critical state up to a local unitary matrix. 
As we will see shortly, almost all states belong to an orbit of some critical state.

The Kempf-Ness theorem~\cite{KN} in the context of this paper says (the only hard
part of the theorem is the \textquotedblleft if\textquotedblright\ part of 3):

\begin{theorem}\label{KN}
Let $\phi\in\mathcal{H}_{n}$ and let $K\equiv SU(d_1)\otimes
SU(d_2)\otimes\cdots\otimes SU(d_n)$, then

\textbf{1.} $\phi\in Crit(\mathcal{H}_{n})$ if and
only if $\left\Vert g\phi\right\Vert \geq\left\Vert \phi\right\Vert $ for all
$g\in G$.\\

\textbf{2.} If $\phi\in Crit(\mathcal{H}_{n})$ and $g\in G$ then $\left\Vert g\phi\right\Vert
\geq\left\Vert \phi\right\Vert $ with equality if and only if $g\phi\in K\phi$.
Moreover, if $g$ is positive definite then the equality condition holds if and only if
$g\phi=\phi$.\\

\textbf{3.} If $\phi\in\mathcal{H}_{n}$ then $G\phi$ is closed in $\mathcal{H}_{n}$ if
and only if $G\phi\cap Crit(\mathcal{H}_{n})\neq\emptyset$.
\end{theorem}

The dimension of the set $Crit(\mathcal{H}_{n})$ in the case of $n$ qubits (i.e. $d_1=d_2=...=d_n=2$) is
$1$ for $n=3$ (up to a local unitary matrix, only the GHZ state belongs to $Crit(\mathcal{H}_{n})$) and $2^n-3n$ for $n>3$.

The following theorem says that critical states are maximally entangled. 
\begin{theorem}\cite{Ver03,GW10}
Let $\psi\in\mathcal{H}_n$. Then,
$
\psi\in Crit(\mathcal{H}_{n})
$
if and only if the \emph{local} density matrices are all proportional to the identity (i.e.
each qudit is maximally entangled with the rest of the system).
\end{theorem}

If an SLOCC orbit $\mathcal{O}_{\psi}$ contains a critical state $\psi_c$, then $\psi_c$ can be regarded as the maximally entangled state of the orbit $\mathcal{O}_{\psi}$. In such cases, we will take $\psi_c$ to represent the orbit and denote the orbit
by $\mathcal{O}_{\psi_c}$. The corollary below implies that almost all SLOCC classes can be represented by a critical state.

\begin{corollary}~\cite{W}
The set of all states of the form $g\psi/\|g\psi\|$, where $g\in G$ and $\psi\in Crit(\mathcal{H}_n)$, 
is dense in $\mathcal{H}_n$.
\end{corollary}

Moreover, part 2 of the Kempf-Ness theorem above implies the uniqueness of critical states; that is, up to local unitary matrix there can be at most one critical state in any SLOCC class $\mathcal{O}_{\psi}$. That is, critical states are the natural representatives of SLOCC orbits, and are the unique maximally entangled states in their SLOCC orbits.

The main results of this paper are expressed in terms of the stabilizer group of the state representing the SLOCC orbit, and
the multipartite generalization of the reduced density matrices, which we call
the associate density matrices (ADMs). In the next two sections we define and prove useful theorems about stabilizer groups and ADMs. These will be fruitful when we present our main results.

\section{The Stabilizer Group}\label{sec:stab}

In section~\ref{sec:main} we will see that the stabilizer group naturally arrises in the analysis of separable operations.
However, unlike the stabilizer formalism of quantum error correcting codes~\cite{Nie00}, in general, the stabilizer group considered in this paper is not a subgroup of the Pauli group. Instead, it is a subgroup of $\tilde{G}$, which is a much bigger group than the Pauli group.  

\begin{definition}
The \emph{stabilizer} group of $\psi\in\mathcal{H}_n$ is a subgroup of $\tilde{G}$
defined by
$$
\text{Stab}(\psi)\equiv\{g\in \tilde{G}\big|\;g\psi=\psi\}\;.
$$
\end{definition}
Note that for $g\in \tilde{G}$ we have $\text{Stab}(g\psi)=g\text{Stab}(\psi)g^{-1}$. That is, the stabilizer groups of two states in the same SLOCC class are related to each other by a conjugation.

\begin{proposition}\label{uni}
If the stabilizer group $Stab(\psi)$, for $\psi\in\mathcal{H}_n$, is finite, then there exists
a state $\phi\in\mathcal{O}_{\psi}$ such that $Stab(\phi)\subset U(d_1)\otimes U(d_2)\otimes\cdots\otimes U(d_n)$.
\end{proposition}

\begin{proof}
The proof is by construction.
Denote by $Stab(\psi)\equiv\{S_k\}_{k=1,2,...,m}$, where $m$ is the number of elements in $Stab(\psi)$.
Since $Stab(\psi)$ is a subgroup of $\tilde{G}$, it follows that the elements $S_k$ can be written as:
$$
S_k=A_{k}^{(1)}\otimes A_{k}^{(2)}\otimes\cdots\otimes A_{k}^{(n)}\;,
$$
where $A_{k}^{(l)}\in\text{GL}(d_l,\mathbb{C})$ and $l=1,2,...,n$. Clearly, for a fixed $l$, the set $\{A_{k}^{(l)}\}_{l=1,2,...,n}$
form a subgroup of $\text{GL}(d_l,\mathbb{C})$. This subgroup is equivalent (up to a conjugation) to a unitary group. To see that, denote,
$$
\Delta^{(l)}\equiv\left(\sum_{k=1}^{m}A_{k}^{(l)\dag}A_{k}^{(l)}\right)^{1/2}\;.
$$
Clearly, $\Delta^{(l)}\in\text{GL}(d_l,\mathbb{C})$ and the elements $U_{k}^{(l)}\equiv \Delta^{(l)}A_{k}^{(l)}(\Delta^{(l)})^{-1}$
form a unitary representation of the subgroup $\{A_{k}^{(l)}\}$. Now, denote by
\begin{align}
& \Delta\equiv \Delta^{(1)}\otimes\Delta^{(2)}\otimes\cdots\otimes \Delta^{(n)}\nonumber\\
& U_k\equiv U_{k}^{(1)}\otimes U_{k}^{(2)}\otimes\cdots\otimes U_{k}^{(n)}\;.
\end{align}
By construction we have
$$
U_k\Delta|\psi\rangle=\Delta|\psi\rangle\;.
$$
Hence, the stabilizer group of $\phi\equiv \Delta|\psi\rangle/\|\Delta|\psi\rangle\|$ consists of unitaries.
\end{proof}
Note that for generic states of 4 or more subsystems~\cite{W} the stabilizer group is finite and therefore the proposition above can be applied to most (i.e. a dense subset) of the states in $\mathcal{H}_n$. We now show that if the orbit of a state with finite stabilizer contains a critical state
then the stabilizer of the critical state is unitary. Since most orbits contains a critical state, it is most natural to choose the critical states to represent the SLOCC orbit (see Fig.1).

\begin{proposition}\label{nolan1}
If the stabilizer group $Stab(\psi)$ for $\psi\in\mathcal{H}_n$ is finite, and there exists
a critical state $\psi_c$ in the orbit $\mathcal{O}_{\psi}$, then $Stab(\psi_c)\subset U(d_1)\otimes U(d_2)\otimes\cdots\otimes U(d_n)$.
\end{proposition}

\begin{proof}
Let $g\in Stab(\psi_c)\subset\tilde{G}$; i.e. $g|\psi_c\rangle=|\psi_c\rangle$. In its polar decomposition $g=up$, where 
$u\in U(d_1)\otimes U(d_2)\otimes\cdots\otimes U(d_n)$ and $p\in G$ is a \emph{positive} definite matrix. Note that 
$\|\psi_c\|=\|g\psi_c\|=\|p\psi_c\|$.
From theorem~\ref{KN} (Kempf-Ness), $|\psi_c\rangle$ being critical implies that $\|p\psi_c\|\geq\|\psi_c\|$ with equality if and only if
$p|\psi_c\rangle=|\psi_c\rangle$. 
This implies that $p\in Stab(\psi_c)$. Moreover, since $\psi_c\in\mathcal{O}_{\psi}$ there exists $h\in G$ such that
$Stab(\psi_c)=hStab(\psi)h^{-1}$. From proposition~\ref{uni}, $Stab(\psi)$ is itself conjugate to a unitary subgroup.
Thus, $p$ is similar to a unitary matrix. However, $p$ is positive definite matrix and since it is similar to a unitary matrix it must be an identity.
\end{proof}

Even though for most states in $\mathcal{H}_n$ the stabilizer is a finite group, for many important states such as GHZ states, and other graph states, the stabilizer group is not finite. For such non-generic states we show that the stabilizer can not even be compact.

\begin{proposition}\label{noncom}
Let $\phi\in\mathcal{H}_{n}$ be such that $Stab(\phi)$ is compact then
$Stab(\phi)$ is finite.
\end{proposition}

\begin{proof}
The stabilizer of an element is a linear algebraic group over 
$\mathbb{C}$. The identity component of this group is therefore also a linear algebraic
group over $\mathbb{C}$. The only connected, linear algebraic group that is compact is the group
with one element.
\end{proof}

When the stabilizer group of a state $\psi\in\mathcal{H}_n$ is not finite and therefore non-compact, it would be very useful
to work with a compact subgroup $T\subset Stab(\psi)$, defined as the intersection of $Stab(\psi)$ with 
$U(d_1)\otimes U(d_2)\otimes\cdots\otimes U(d_n)$. We will discuss it in more details in Sec~\ref{sec:conti} and Sec~\ref{sec:ex}. 

\subsection{The Stabilizer Twirling Operation}

In this subsection we define the stabilizer twirling operation that will be used quite often when we discuss our main results
in the next section.

\begin{definition}
Let $\psi\in\mathcal{H}_n$ be a state with a finite unitary stabilizer. The $Stab(\psi)$-twirling operation, $\mathcal{G}(\cdot)$, is defined by
\begin{equation}\label{gtwirling}
\mathcal{G}(\sigma)\equiv\frac{1}{m}\sum_{k=1}^{m}U_{k}^{\dag}\sigma U_{k}\;,
\end{equation}
for any positive semi-definite operator $\sigma$ acting on $\mathcal{H}_n$.
\end{definition}

In Ref.~\cite{BRS07} it was shown that 
the twirling operation $\mathcal{G}$ can be factorized in terms of the irreps of the group (in our case
$Stab(\psi)$). In the following we make use of techniques introduced in Ref.~\cite{BRS07}, to which the reader is referred
for more details. Under the action of the unitary representation of a
compact group $Stab(\psi)$, a finite dimensional Hilbert space factorizes as follows
\begin{equation*}
\mathcal{H}=\sum_{q}\mathcal{H}_{q}=\sum_{q}\mathcal{M}_{q}\otimes \mathcal{N}_{q}
\end{equation*}
where $q$ labels the irreps of $Stab(\psi)$, $\mathcal{M}_{q}$ is the $q$th representations space, and $\mathcal{N}_{q}$
is the $q$th muliplicity space.
The $Stab(\psi)$-twirling operation has the form
\begin{equation}\label{factorization}
\mathcal{G}\left( \sigma \right) =\sum_{q}\mathcal{D}_{\mathcal{M}_{q}}\otimes \mathrm{id}_{\mathcal{N}_{q}}\left( \Pi _{\mathcal{H}_{q}}\sigma \Pi _{\mathcal{H}_{q}}\right) ,
\end{equation}
where $\Pi _{\mathcal{H}_{q}}$ is the projector onto $\mathcal{H}_{q},$ $\mathrm{id}_{\mathcal{N}_{q}}$ is the identity map on 
$\mathcal{N}_{q},$ and $\mathcal{D}_{\mathcal{M}_{q}}$ is the completely decohering map on $\mathcal{M}_{q}$; i.e. it denotes the trace-preserving operation that takes every operator on the Hilbert space $\mathcal{M}_q$ to a constant times the identity. In section~\ref{sec:niel} we will see that this factorization enables the elegant classification of all states
$\phi$ to which $\psi$ can be transformed deterministically by separable operations.

As we pointed out above, if the stabilizer group is not finite, then it is also non-compact and therefore there is no 
$Stab(\psi)$-twirling. However,  for the compact subgroup, $T\subset Stab(\psi)$ as defined above, there is an analog for the twirling operation. We extend the definition of the stabilizer twirling operation
(see Eq.(\ref{gtwirling})) to the compact group $T$.  
\begin{definition}
Let $\psi\in\mathcal{H}$ and assume $Stab(\psi)$ is non-compact. Let $T$ be the intersection of $Stab(\psi)$ with $U(d_1)\otimes U(d_2)\otimes\cdots\otimes U(d_n)$. Then, for any positive semi-definite operator $\sigma$ acting on $\mathcal{H}_n$, the $T$-twirling operation, $\mathcal{T}(\sigma)$, is defined by
\begin{equation}\label{gtwirling2}
\mathcal{T}(\sigma)\equiv\int dt\; t^{\dag}\sigma t\;,
\end{equation}
where $t\in T$ and $dt$ is the Haar measure over $T$.
\end{definition}

The factorization form in Eq.(\ref{factorization}) holds for the $T$-twirling as well. This $T$-twirling operation will become useful whenever the stabilizer group is non-compact. The simplest example of that is the 3 qubits GHZ state (see sec~\ref{sec:ex}). 

\subsection{SL- Invariant Polynomials}\label{SL}

In this subsection we discuss the quantification of entanglement in terms of $SL$-invariant polynomials and discuss their relationship to the stabilizer group. 

An $SL$-invariant polynomial, $f(\psi)$, is a polynomial in the components of the vector 
$\psi\in\mathcal{H}_n$, which is invariant under the action of the group $G$. That is, $f(g\psi)=f(\psi)$ for all $g\in G$. 
In the case of two qubits there exists only one unique $SL$-invariant polynomial.
It is homogeneous of degree $2$ and is given by the bilinear form $(\psi,\psi)$:
$$
f_2(\psi)\equiv (\psi,\psi)\equiv \langle\psi^{*}|\sigma_y\otimes\sigma_y|\psi\rangle\;\;,\;\;\psi\in\mathbb{C}^{2}\otimes\mathbb{C}^2\;.
$$
Its absolute value is the celebrated concurrence~\cite{Woo98}. 

Also in three qubits there exists a unique $SL$-invariant polynomial. It is homogeneous of degree $4$ and
is given by
$$
f_4(\psi)=\det\left[\begin{array}{cc} (\varphi_0,\varphi_0) & (\varphi_0,\varphi_1)\\ 
(\varphi_1,\varphi_0) &(\varphi_1,\varphi_1)\end{array}\right]\;,
$$
where the two qubits states $\varphi_i$ for $i=0,1$ are defined by the decomposition 
$|\psi\rangle=|0\rangle|\varphi_0\rangle+|1\rangle|\varphi_1\rangle$, and the bilinear form $(\varphi_i,\varphi_j)$ is defined above for two qubits.
The absolute value of $f_4$ is the
celebrated 3-tangle~\cite{CKW}. 

In four qubits, however, there are many $SL$-invariant polynomials and it is possible to show
that they are generated by 4 SL-invariant polynomials (see e.g.~\cite{GW10} for more details and refrences).

Any $SL$-invariant polynomial is generated by a finite number of homogeneous $SL$-invariant polynomials.
Moreover, for some states, all $SL$-invariant polynomials vanish. For example, the $W$ state 
$$
|W\rangle=(|100\rangle+|010\rangle+|001\rangle)/\sqrt{3}
$$ 
and (its generalization to $n$-qubits) has the property that $g_t|W\rangle=t|W\rangle$, where 
$g_t\equiv\text{diag}\{t,t^{-1}\}^{\otimes 3}$ and $t\in\mathbb{R}$ is non-zero. Note that $g_t\in SL(2,\mathbb{C})^{\otimes 3}$.
Thus, the value of any homogeneous $SL$-invariant polynomial, $f$, is zero on $|W\rangle$ (and its generalization to $n$ qubits) since 
$$
f(|W\rangle)=f(g_t|W\rangle)=f(t|W\rangle)=t^kf(|W\rangle),
$$   
where $k\in\mathbb{N}$ is the degree of $f$. Nevertheless, the set of states for which \emph{all} $SL$-invariant polynomials vanish
is of measure zero and is called here the \emph{null cone}. Since the null cone is of measure zero, for most states there exists at least one homogeneous $SL$-invariant polynomial that does not vanish on them. The following proposition will be useful for such states.

\begin{proposition}\label{gk}
Let $\psi\in\mathcal{H}_n$ and assume $\psi$ is not in the null cone.
Let $k\in\mathbb{N}$ be the degree of some homogeneous SL-invariant polynomial 
that is not zero on $\psi$. Then,
$$
Stab(\psi)\subset G_k
$$
where
$$
G_k\equiv\left\{g'\in\tilde{G}\big| g'=e^{i\frac{2\pi m}{k}}g\;;\; 1\leq m\leq k\;;\;g\in G\right\}\;
$$
is a subgroup of $\tilde{G}$. 
\end{proposition}

\begin{proof}
Let $f$ be an homogeneous SL-invariant polynomial of degree $k$ such that $f(\psi)\neq 0$.
Let $g'$ be an element in $Stab(\psi)$. Since $Stab(\psi)\subset\tilde{G}$, there exists $a\in\mathbb{C}$ such that
$g'=ag$, where $g\in G$. 
Thus,
\begin{align*}
f(\psi) =f(g'\psi)=f\left(ag\psi\right)
=a^kf\left(g\psi\right)
=a^kf\left(\psi\right)\;,
\end{align*} 
where the first equality follows from the fact that $g'\in Stab(\psi)$, whereas the last equality follows
from the fact that $g\in G$ and $f$ is $SL$-invariant polynomial.
Since $f(\psi)\neq 0$ it follows that $a^k=1$. This completes the proof.
\end{proof}

\begin{remark}
If $k$, in the proposition above, is equal to the dimension, $d_l$, of the subsystem $l$ (here $l=1,2,...,n$),
then $G_k=G$. This is because $e^{i2\pi m/d_l}g\in G$ for all $g\in G$, $l=1,2,...,n$ and $m=1,2,...,d_l$.
In particular, for systems of $n$ qubits $G_2=G$.
\end{remark}

In the following corollary, we show that for most states in $\mathcal{H}_n$ (i.e. a dense subset of $\mathcal{H}_n$)
the stabilizer group is a subgroup of $G$ (which is a smaller and somewhat simpler group to work with than $\tilde{G}$).

\begin{corollary}
Let $f_k$ and $f_{k'}$ be two homogenous $SL$-invariant polynomials of degrees $k$ and $k'$, respectively. 
Let $\psi\in\mathcal{H}_n$, and assume that $f_k(\psi)\neq 0$ and $f_{k'}(\psi)\neq 0$. Denote by $l=(k,k')$ the greatest common divisor (gcd) of $k$ and $k'$. Then 
$$
Stab(\psi)\subset G_l\;.
$$
\end{corollary}
This corollary follows directly from the proposition above. Note also that
since most states in $\mathcal{H}_n$ has many non-vanishing SL-invariant homogenous polynomials, it follows from the corollary above that for these states the stabilizer group is a subgroup of $G$. For example, consider the case of $n$-qubits with even $n$.
For such systems their exists a homogeneous $SL$-invariant polynomial of degree 2. It is possible to show that for a dense set of states in $\mathcal{H}_n$ this polynomial has non-vanishing value. Thus, for this dense set of states the stabilizer is a subgroup of
$G_2$ which is equal to $G$ for qubits.

Despite the above corollary, in many important cases (such as the 3 qubits GHZ state) 
the stabilizer is not a subgroup of $G$. For these cases the following corollary will be useful. 
\begin{corollary}\label{stst}
Let $\psi\in\mathcal{H}_n$ and define $St(\psi)$ to be the intersection of $Stab(\psi)$ with $G$; that is,
$$
\text{St}(\psi)\equiv\{g\in G\big|\;g\psi=\psi\}\;.
$$
If there exists an homogeneous SL-invariant polynomial of degree $k$ that is not vanishing on $\psi$,
then the quotient group $Stab(\psi)/St(\psi)$ is a cyclic group of order at most $k$.
\end{corollary} 
This corollary is a direct consequence of proposition~\ref{gk}.

\section{The Associate Density Matrix}\label{sec:asso}

In this section we associate density matrices (not necessarily normalized) with states in an orbit $\mathcal{O}_{\psi}$.
We call such density matrices \emph{associate density matrices} (ADMs). We will see that the ADM plays a similar role
to that of the reduced density matrices, when considering local transformations that involves pure bipartite states.
\begin{definition}\label{asso}
Let $\psi\in\mathcal{H}_n$ and let $\phi\in\mathcal{O}_{\psi}$ be another state in the SLOCC orbit of $\psi$.
Thus, there exists $g\in G$ such that $\phi=g\psi/\|g\psi\|$. An \emph{associate} density matrix, $\rho_{\psi}(\phi)$, of $\phi$ with respect to $\psi$, is defined by
$$
\rho_{\psi}(\phi)=\frac{g^{\dag}g}{\|g\psi\|^2}\;,\;\text{where } \phi=\frac{g\psi}{\|g\psi\|}\;.
$$
\end{definition}

Note that the associate density matrix is positive definite and that for $s\in\text{Stab}(\psi)$, both $\rho_{\psi}(\phi)$ and $s^{\dag}\rho_{\psi}(\phi)s$ are associate
density matrices of $\phi$ with respect to $\psi$. Hence, $\rho_{\psi}(\phi)$ is defined up to a conjugation by
the stabilizer group of $\psi$. Furthermore, note that $\rho_{\psi}(\phi)= I$ if and only if $\psi$ is related
to $\phi$ by a local unitary transformation. 

\begin{figure}[tp]
\includegraphics[scale=.40]{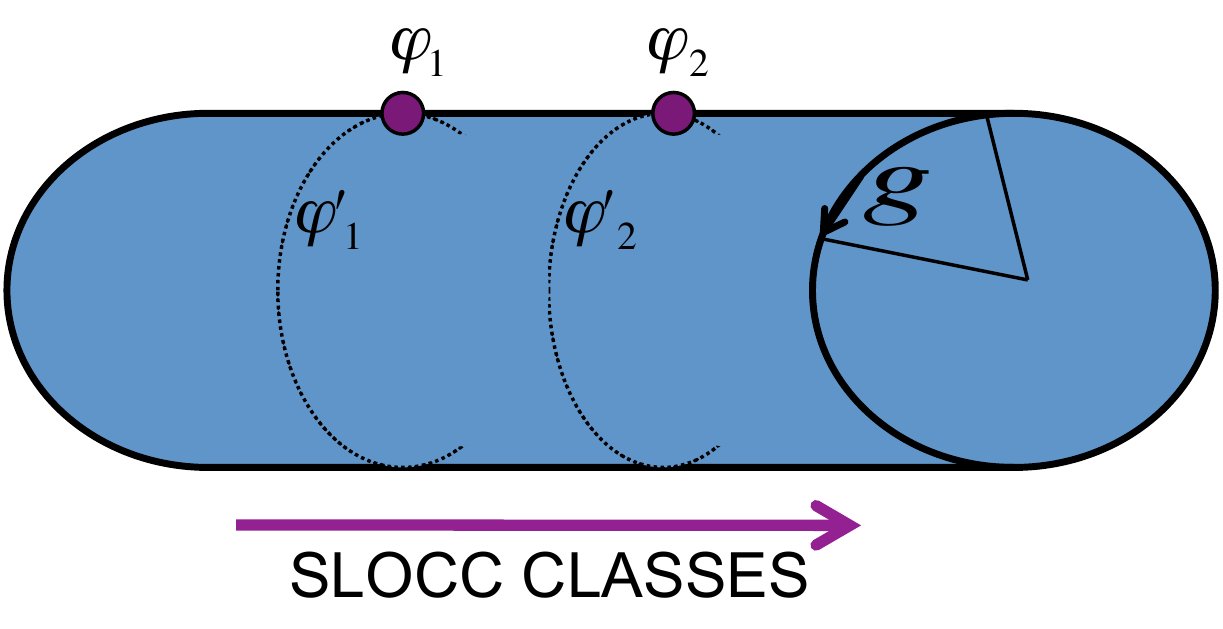}
\caption{The state space $\mathcal{H}_n$ can be described heuristically as a horizontal cylinder. The horizontal dimension represents the uncountable number of inequivalent SLOCC classes. Each SLOCC orbit of states is described by a circle perpendicular to the horizontal line. The states $\varphi _1$ and $\varphi_2$ are the representatives of two inequivalent SLOCC orbits. The action of $g\in G$ on $\varphi _1$ and $\varphi_2$ generates $\varphi _{1}'$ and $\varphi_{2}'$, respectively. Therefore, the ADMs of both $\varphi _{1}'$ and $\varphi_{2}'$ are proportional to $g^{\dag}g$. 
From the results in section~\ref{sec:niel} it follows that if the stabilizers of $\varphi _1$ and $\varphi_2$ are equal, then the transformation $\varphi _1\rightarrow\varphi _{1}'$ is possible by SEP iff $\varphi _2\rightarrow\varphi _{2}'$ is possible by SEP.}
\end{figure}

\begin{proposition}\label{proasso}
Let $\psi,\phi\in\mathcal{H}_n$ be two states in the same orbit $\mathcal{O}_{\psi}$ (i.e. $\phi=g\psi/\|g\psi\|$, for
some $g\in G$), where $Stab(\psi)\subset U(d_1)\otimes U(d_2)\otimes\cdots\otimes U(d_n)$. Then,
$Stab(\phi)\subset U(d_1)\otimes U(d_2)\otimes\cdots\otimes U(d_n)$ if and only if
$$
[\rho_{\psi}(\phi),u]=0\;\;\text{for all  }u\in Stab(\psi)\;.
$$
\end{proposition}

\begin{proof}
Since $\phi=g\psi/\|g\psi\|$ we have $Stab(\phi)=g Stab(\psi)g^{-1}$. That is, if $s\in Stab(\phi)$ then there exists
$u\in Stab(\psi)$ such that $s=gug^{-1}$. Thus, $s$ is unitary if and only if 
$$
s^{\dag}s=(g^{\dag})^{-1}u^{\dag}g^{\dag}gug^{-1}=1\;\iff\;u^{\dag}g^{\dag}gu=g^{\dag}g\;.
$$
This completes the proof.
\end{proof}

\subsection*{The relation with the reduced density matrix}

Consider the Hilbert space $\mathbb{C}^{d}\otimes\mathbb{C}^{d}$, and the maximally entangled state
(it is unnormalized for simplicity)
$$
|\psi\rangle=\sum_{i=0}^{d-1}|i\rangle |i\rangle\;,
$$
where $\{|i\rangle\}$ is an orthonormal basis in $\mathbb{C}^d$. Clearly, in this case the orbit $\mathcal{O}_{\psi}$
consists of all states in $\mathbb{C}^{d}\otimes\mathbb{C}^{d}$ with Schmidt number $d$. 

The stabilizer group of $\psi$ is given by (see Sec.~\ref{bip} for the proof):
\begin{equation}\label{bistab}
Stab(\psi)=\left\{S^{-1}\otimes S^{T} \Big| \;S\in GL(\mathbb{C}^{d})\right\}\;,
\end{equation}
where $S^{T}$ is the transpose matrix of $S$. 
 
Any state $\phi\in\mathcal{O}_{\psi}$
can be written as
$$
|\phi\rangle=\frac{1}{\|A_1\otimes A_2|\psi\rangle\|}A_1\otimes A_2|\psi\rangle\;,
$$
where $A_1\otimes A_2\in G$ and in this subsection $G=SL(d,\mathbb{C})\otimes SL(d,\mathbb{C})$.
However, due to the relatively big stabilizer group, we can always assume
without loss of generality that $A_2=I$. Thus, the associate density matrix (which is defined up to conjugation by an element from the stabilizer group) can be written as:
$$
\rho_{\psi}(\phi)=\frac{A_{1}^{\dag}A_1\otimes I}{\|A_1\otimes I|\psi\rangle\|^2}=\rho_r\otimes I\;,
$$
where $\rho_r\equiv\text{Tr}_{B}|\phi\rangle\langle\phi|$ is the reduced density matrix. Thus, in the bipartite case, the
associate density matrix is essentially equivalent to the reduced density matrix. As we will see in the next section, in the multipartite case, the associate density matrix plays the same role as the reduced density matrix plays in bipartite LOCC transformations.

\section{Main Results}\label{sec:main}

In this section we present our main results. As we have seen in the previous sections,
for a dense set of states in $\mathcal{H}_n$ the stabilizer group is finite. Hence, in the first three subsections we assume explicitly that the stabilizer is finite, while our results for infinite stabilizer groups are presented in the last subsection.
More specifically, in subsection~\ref{sec:niel} we generalize Nielsen majorization theorem to multi-partite states, and in subsection~\ref{NonD}
we introduce necessary and sufficient conditions for non-deterministic SEP transformations. Subsection~\ref{sec:pro} is devoted to the calculation of the maximum probability to convert one state to another in the case where deterministic SEP is not possible. Finally, in section~\ref{sec:conti} we generalize our theorems for the case in which the stabilizer group is non-compact.

\subsection{Generalization of Nielsen Majorization Theorem to Multipartite states}\label{sec:niel}
 
 In~\cite{Nie99} it was shown that a transformation $|\psi_1\rangle^{\text{AB}}\to|\psi_2\rangle^\text{AB}$, between two bipartite states in $\mathbb{C}^{n}\otimes \mathbb{C}^{m}$, is possible by LOCC if and only if there exists a set of probabilities $\{p_k\}$ and a set of unitary matrices $\{U_k\}$ such that 
 $$
 \rho_{1}^{\text{A}}=\sum_{k} p_kU_{k}^{\dag}\rho_{2}^{A}U_{k}\;,
 $$
 where $\rho_{i}^{A}\equiv\text{Tr}_\text{B}|\psi_i\rangle^{\text{AB}}\langle\psi_i|$ ($i=1,2$) are the reduced density matrices.
 Moreover, it was argued that the relation in the equation above hold true if and only if the eigenvalues of $\rho_{1}^{\text{A}}$
 are majorized by the eigenvalues of $\rho_{2}^{\text{A}}$; i.e. $\rho_{1}^{\text{A}}\prec\rho_{2}^{\text{A}}$. We now generalize this theorem to multipartite states.
 
\begin{theorem}\label{two}
Let $|\psi_1\rangle=g_1|\psi\rangle/\|g_1\psi\|$ and $|\psi_2\rangle=g_2|\psi\rangle/\|g_2\psi\|$
be two entangled states in the same orbit $\mathcal{O}_{\psi}$, where the orbit representative 
$|\psi\rangle\in\mathcal{H}_n$ is chosen such that $Stab(|\psi\rangle)$ is a finite unitary group 
(see proposition~\ref{uni}). Let $m$ be the cardinality of $Stab(\psi)$.
Then, $|\psi_1\rangle\rightarrow|\psi_2\rangle$ by SEP 
if and only if there exists a set of probabilities $\{p_k\}$ such that
\begin{equation}\label{basic1}
\sum_{k=1}^{m}p_kU_{k}^{\dag}\rho_2U_{k}=\rho_1\;,
\end{equation}
where $\rho_i\equiv\rho_{\psi}(\psi_i)$ for $i=1,2$ and $U_k\in Stab(\psi)$.
\end{theorem}

\begin{remark}
The condition on $\rho_1$ and $\rho_2$ in the Eq.(\ref{basic1}) is the multipartite generalization of the condition obtained in~\cite{Nie99} for the bipartite case. In particular, Eq.(\ref{basic1}) implies that $\rho_1\prec\rho_2$, i.e. the eigenvalues of $\rho_1$ are majorized by the eigenvalues of $\rho_2$. In the bipartite case, this majorization condition is both necessary \emph{and} sufficient. The reason for this is that in the bipartite case the stabilizer group is big and contains all unitary matrices, whereas in the multipartite state the stabilizer usually consists of only a finite number of unitary matrices $U_k$. Hence, in the multipartite state the majorization
condition is far from being sufficient.
\end{remark}

\begin{proof}
Consider the separable operation given in Eq.(\ref{sep}). If $\Lambda(|\psi_1\rangle\langle\psi_1|)=|\psi_2\rangle\langle\psi_2 |$ then
there must exist complex coefficients $a_k$ such that
$$
M_k|\psi_1\rangle=a_k|\psi_2\rangle
$$
for all $k$. Since the separable operation $\Lambda$ is invariant under the multiplication of the operators $M_k$ by phases,
we can assume without loss of generality that $a_k$ are real and non-negative. We therefore denote $a_k\equiv\sqrt{p_k}$,
where $p_k$ are non-negative real numbers. From the above equation and the completeness relation
$\sum_kM_{k}^{\dag}M_k=I$ it follows that $\sum_kp_k=1$.

Now, since $|\psi_1\rangle=g_1|\psi\rangle/\|g_1\psi\|$ and $|\psi_2\rangle=g_2|\psi\rangle/\|g_2\psi\|$ we get
$$
\frac{\|g_2\psi\|}{\|g_1\psi\|\sqrt{p_k}}g_{2}^{-1}M_kg_{1}|\psi\rangle=|\psi\rangle\;.
$$
Hence,
$$
\frac{\|g_2\psi\|}{\|g_1\psi\|\sqrt{p_k}}g_{2}^{-1}M_kg_{1}=U_k\;,
$$
where $U_k\in Stab(\psi)$. In the above equation we used the fact that $\psi$ is n-way entangled (in the sense that the determinant
of each of the $n$ reduced density matrices obtained by tracing out $n-1$ qudits, is non zero); otherwise,
$M_k$ may have zero determinant. The above equation may be rewritten as: 
\begin{equation}\label{mkform}
M_k=\frac{\|g_1\psi\|}{\|g_2\psi\|}\sqrt{p_k}g_2U_kg_{1}^{-1}\;.
\end{equation}
The completeness relation $\sum_kM_{k}^{\dag}M_k=I$ gives:
$$
\sum_{k=1}^{m}p_kU_{k}^{\dag}\rho_{2}U_{k}=\rho_1\;,
$$
where we substitute $\rho_i=g_{i}^{\dag}g_{i}/\|g_{i}\psi\|^2$ for $i=1,2$. This provide the first direction of the proof.
The other direction follows the exact same lines backwards.
\end{proof}

Taking on both sides of Eq.(\ref{basic1}) the $Stab(\psi)$-twirling operation gives the following simple necessary condition
which we summarize in the following corollary:
\begin{corollary}
Notations as in theorem~\ref{two}. If the transformation $|\psi_1\rangle\to|\psi_2\rangle$ is achievable by SEP then
\begin{equation}\label{nece}
\mathcal{G}\left(\rho_2\right)=\mathcal{G}\left(\rho_1\right)\;.
\end{equation}
Moreover, if $\rho_1$ is invariant under the action of the $Stab(\psi)$-twirling operation, that is, if
$$
\mathcal{G}\left(\rho_1\right)=\rho_1\;,
$$
then the condition in Eq.(\ref{nece}) is both necessary and sufficient.
\end{corollary}

The corollary above implies in particular that the transformation $\psi\to\psi_2$ is achievable by SEP if and only if
\begin{equation}\label{abc}
\mathcal{G}\left(\rho_2\right)=I\;.
\end{equation}

We have seen in the previous sections that if $\psi\in\mathcal{H}_n$ is both critical and generic, then 
$Stab(\psi)=\{U_k\}_{k=1,2,...,m}$
is a finite unitary group. Hence, from Eq.~(\ref{abc}) and Eq.~(\ref{factorization}) we can
characterize all the states $\phi$ that can be obtained from a critical generic state $\psi$.
We summarize this characterization in the following corollary:
\begin{corollary}
Let $\psi\in\mathcal{H}_n$ be a state with a unitary stabilizer group, and let $\phi$ be a state in the orbit $\mathcal{O}_{\psi}$. Then,
the transformation $\psi\rightarrow\phi$ can be achieved by SEP  
if and only if
$$
\text{Tr}_{\mathcal{M}_q}\left[\Pi _{\mathcal{H}_{q}}\rho_{\psi}(\phi) \Pi _{\mathcal{H}_{q}}\right]=\mathrm{id}_{\mathcal{N}_{q}}\;\;\;\forall\;q\;.
$$
\end{corollary}

\subsection{Non Deterministic Transformations}\label{NonD}

Transformations involving quantum measurements are usually not deterministic. 
In this subsection we provide necessary and sufficient conditions for the existence of a local procedure that converts a finite multipartite pure state into one of a set of possible final states. For the bipartite case, this problem was solved in~\cite{Ple99}. 
In particular, a minimal set of necessary and sufficient conditions
where found for the local conversion of a bipartite state $|\psi\rangle^{AB}$ to the ensemble of bipartite states $\{p_j,\;|\phi_j\rangle^{AB}\}$;
i.e., $|\psi\rangle^{AB}$ is converted to $|\phi_j\rangle^{AB}$ with probability $p_j$. The conditions are
$$
E_{k}\left(|\psi\rangle^{AB}\right)\geq\sum_{j}p_jE_{k}\left(|\phi_j\rangle^{AB}\right)
$$
where $\{E_k\}$ are the Vidal's entanglement monotones~\cite{Vid99}:
\begin{equation}\label{Vid}
E_{k}\left(|\psi\rangle^{AB}\right)\equiv\sum_{m=k}^{n}\lambda_m\;\;,\;\text{for}\;k=1,2,...,n
\end{equation}
where $\lambda_1\geq\lambda_2\geq\cdots\geq\lambda_n$ are the ordered eigenvalues of the reduced density matrix 
$\rho^A\equiv\text{Tr}_{B}|\psi\rangle^{AB}\langle\psi|$.
The following theorem generalizes these conditions to the multipartite case.

\begin{theorem}\label{basic2}
Let $\phi\in\mathcal{H}_n$ be a normalized state in the orbit $\mathcal{O}_{\psi}$, where the orbit representative $|\psi\rangle$
is chosen such that $Stab(\psi)$ consists of $m$ unitary matrices (see proposition~\ref{uni}).  Let $\{p_j,\phi_j\}$ be an ensemble of states also in the orbit $\mathcal{O}_{\psi}$. Then, the transformation $\phi\rightarrow\{p_j,\phi_j\}$ is achievable  
by SEP (i.e. there exists a separable operation taking $\psi$ to $\phi_j$ with probability $p_j$) if and only if there exist probabilities 
$\{p_{jk}\}$ such that:
\begin{align}\label{ggg}
&\sum_{k=1}^{m}p_{jk}=p_j\nonumber\\
&\sum_{j}\sum_{k=1}^{m}p_{jk}U_{k}^{\dag}\rho_jU_{k}=\rho\;,
\end{align} 
where $\rho\equiv\rho_{\psi}(\phi_j)$, $\rho_j\equiv\rho_{\psi}(\phi_j)$, and $U_{k}\in Stab(\psi)$.
\end{theorem}

\begin{remark}
The relation between $\rho$ and $\rho_j$ in Eq.(\ref{ggg}) is the multipartite generalization of the condition obtained in~\cite{Ple99} for the bipartite case. In particular, Eq.(\ref{ggg}) implies that 
\begin{equation}\label{ensvid}
E_{k}\left(|\phi\rangle\right)\geq\sum_{j}p_jE_{k}\left(|\phi_j\rangle\right)\;,
\end{equation}
for all $k=1,2,...,\dim\mathcal{H}_n$. On multipartite states the functions $\{E_k\}$ are defined by
\begin{equation}\label{gvid}
E_{k}\left(|\phi\rangle\right)\equiv\sum_{m=k}^{\dim\mathcal{H}_n}\lambda_m\;\;,\;\text{for}\;k=1,2,...,\dim\mathcal{H}_n
\end{equation}
where $\lambda_1\geq\lambda_2\geq\cdots\geq\lambda_{\dim\mathcal{H}_n}$ are the ordered eigenvalues of 
the ADM of $|\phi\rangle$.
In the bipartite case, these conditions are both necessary \emph{and} sufficient since the stabilizer group contains all unitary matrices, whereas in the multipartite case the stabilizer usually consists of only a finite number of unitary matrices $U_k$. Hence, in the multipartite case the conditions in Eq.(\ref{ensvid}) are far from being sufficient.
\end{remark}

\begin{proof}
The proof follows similar lines as the proof of theorem~\ref{two}.
Consider the SEP $\Lambda=\sum_{j}\Lambda_{j}$, where 
$$
\Lambda_{j}(\rho)\equiv\sum_{k=1}^{m}M_{jk}\rho M_{jk}^{\dag}
$$ 
for any density matrix $\rho$ acting on $\mathcal{H}_n$. If the transformation 
$\phi\rightarrow\{p_j,\phi_j\}$ is achievable by such a SEP, then 
$$
\Lambda_{j}(|\psi\rangle\langle\psi|)=p_j|\phi_j\rangle\langle\phi_j|\;.
$$ 
Thus, there exists complex coefficients $a_{jk}$ such that
$$
M_{jk}|\phi\rangle=a_{jk}|\phi_j\rangle
$$
for all $j$ and $k$. Since the separable operation $\Lambda$ is invariant under the multiplication of the operators $M_{jk}$ by phases,
we can assume without loss of generality that $a_{jk}$ are real and non-negative. We therefore denote $a_{jk}\equiv\sqrt{p_{jk}}$,
where $p_{jk}$ are non-negative real numbers. From the above equation it follows that $\sum_kp_{jk}=p_j$.

Now, since $|\phi\rangle=g|\psi\rangle/\|g\psi\|$ and $|\phi_j\rangle=g_j|\psi\rangle/\|g_j\psi\|$ for some $g,g_j\in G$, we get
$$
\frac{\|g_j\psi\|}{\|g\psi\|\sqrt{p_{jk}}}g_{j}^{-1}M_{jk}g|\psi\rangle=|\psi\rangle\;.
$$
Hence,
$$
\frac{\|g_j\psi\|}{\|g\psi\|}\sqrt{p_{jk}}g_{j}^{-1}M_{jk}g=S_k\;,
$$
where $S_k\in Stab(\psi)$. In the above equation we used the fact that $\psi$ is n-way entangled (in the sense that the determinant
of each of the $n$ reduced density matrices, obtained by tracing out $n-1$ qudits, is non zero);
otherwise $M_{jk}$ may have zero determinant. The above equation may be rewritten as: 
$$
M_{jk}=\frac{\|g\psi\|}{\|g_{j}\psi\|}\sqrt{p_{jk}}g_{j}S_kg^{-1}\;.
$$
Thus, the completeness relation $\sum_{jk}M_{jk}^{\dag}M_{jk}=I$ gives:
$$
\sum_{j}\sum_{k=1}^{m}p_{jk}U_{k}^{\dag}\rho_jU_{k}=\rho\;,
$$
where we substitute $\rho\equiv \rho_{\psi}(\phi_j)$ and $\rho_j\equiv g_{j}^{\dag}g_{j}/\|g_j\psi\|^2$. This provide the first direction of the proof.
The other direction follows the exact same lines backwards.
\end{proof}

Taking on both sides of Eq.(\ref{ggg}) the $Stab(\psi)$-twirling operation gives the following simple necessary condition
which we summarize in the following corollary:
\begin{corollary}
Notations as in theorem~\ref{basic2}. If the transformation $|\phi\rangle\to\{p_j,|\phi_j\rangle\}$ is achievable by SEP then
\begin{equation}\label{nece2}
\sum_{j}p_j\mathcal{G}(\rho_j)=\mathcal{G}(\rho)\;.
\end{equation}
Moreover, if $\rho$ is invariant under the action of the $Stab(\psi)$-twirling operation, that is, if
$$
\mathcal{G}\left(\rho\right)=\rho\;,
$$
then the condition in Eq.(\ref{nece2}) is both necessary and sufficient.
\end{corollary}

The corollary above implies in particular that the transformation $\psi\to\{p_j,\phi_j\}$ is achievable by SEP if and only if
$$
\sum_{j}p_j\mathcal{G}(\rho_j)=I\;.
$$

\subsection{Probabilities}\label{sec:pro}

From the theorems above we see that when the conversion of one state to another is possible with some probability, i.e. the states in the same SLOCC class, this probability is usually less than one. A natural question to ask is what is the maximum possible probability, $P_{\max}$, to convert one state to another if we know that the two states are in the same SLOCC class. 
Eq.~(\ref{ensvid}) provides an immediate upper bound on $P_{\max}$, which we summarize in the following corollary.
\begin{corollary}
Let $\psi_1$ and $\psi_2$ be as in theorem~\ref{two}.
Then,
\begin{equation}\label{vidupper}
P_{\max}(\psi_1\to\psi_2)\leq\min_{k}\left\{\frac{E_{k}(\psi_1)}{E_{k}(\psi_2)}\right\}
\end{equation}
where $k=1,2,...,\dim\mathcal{H}_n$ and $E_k$ are the functions defined in Eq.(\ref{gvid}).
\end{corollary}
\begin{proof}
The procedure converting $|\psi_1\rangle$ to $|\psi_2\rangle$ with probability $P_{\max}$ also converts $|\psi_1\rangle$
to some other states with the remaining probability $1-P_{\max}$. The value of $E_k$ for these states is non-zero
and therefore from Eq.(\ref{ensvid}) it follows that $P_{\max}\leq E_{k}(\psi_1)/E_{k}(\psi_2)$.
\end{proof}
In~\cite{Vid99} it was shown that for the bipartite case the upper bound in Eq.(\ref{vidupper}) is always achievable.
In the multipartite case, unless $\psi_2$ is critical, this is usually not the case and as we will see below in some 
cases it is even possible 
to find better upper bounds. In particular, it is possible that for all $k$, $E_{k}(\psi_1)\geq E_{k}(\psi_2)$ and yet the probability 
$P_{\max}$ is strictly less than 1 (or even arbitrarily close to zero).
 
In the following we first derive a mathematical expression for the maximum probability. Since in general the expression can be difficult to calculate (at least as hard as determining if a multipartite density matrix is separable or not), we simplify it dramatically for the case in which the states $\psi_1$ or $\psi_2$ are critical. Then, for the more general case, we find lower and upper bounds for $P_{\max}$. These upper and lower bounds remain the same if we replace SEP by LOCC.

\subsubsection{General formula for $P_{\max}$}

\begin{theorem}\label{delta}
Let $\psi,\;\psi_1,\;\psi_2,\;\rho_1\;,\rho_2$ and $m$ be as in theorem~\ref{two}. 
For any set of probabilities $\{p_k\}_{k=1}^{m}$ with 
$\sum_{k=1}^{m}p_k\leq 1$, denote
$$
\Delta_{\psi_1\to\psi_2}\left(\{p_k\}\right)\equiv \rho_1-\sum_{k=1}^{m}p_kU_{k}^{\dag}\rho_2U_{k}
$$
where $U_{k}\in Stab(\psi)$.
Then, the maximum probability to convert $\psi_1$ to $\psi_2$ is given by the following expression:
\begin{align}
& P_{\max}(\psi_1\to\psi_2)\nonumber\\
& =\max\left\{\sum_{k=1}^{m}p_k\Big| \;\Delta_{\psi_1\to\psi_2}\left(\{p_k\}\right)\text{ is separable}\right\}\;.
\label{opt}
\end{align}
That is, the maximum is taken over all sets of probabilities $\{p_k\}$ such that $\sum_{k=1}^{m}p_k\leq 1$ and $\Delta_{\psi_1\to\psi_2}\left(\{p_k\}\right)$ is a positive semi-definite separable matrix.
\end{theorem}

\begin{proof}
Any separable transformation that takes $\psi_1$ to $\psi_2$ with maximum probability $P_{\max}$, takes $\psi_1$ to some other states that are not $\psi_2$ with the remaining probability $1-P_{\max}$. However, since any state can be transformed to the product state 
$|00...0\rangle$ with 100\% success , we can assume, without loss of generality, that the separable map that takes $\psi_1$ to $\psi_2$ with maximum probability $P_{\max}$, takes $\psi_1$ to the product state $|00...0\rangle$ with probability $1-P_{\max}$. A separable transformation
with such property, can be written as $\Lambda=\Lambda_1+\Lambda_2$, where 
$$
\Lambda_1(\cdot)=\sum_{k=1}^{m}M_{k}(\cdot) M_{k}^{\dag}
$$
is a separable map that takes $\psi_1$ to $\psi_2$ with probability $\text{Tr}\left[\Lambda_1(\psi_1)\right]$ and 
$$
\Lambda_2(\cdot)=\sum_{k>m}M_{k}(\cdot) M_{k}^{\dag}
$$
is a separable map that takes $\psi_1$ to $|00...0\rangle$ with probability $\text{Tr}\left[\Lambda_2(\psi_1)\right]=1-\text{Tr}\left[\Lambda_1(\psi_1)\right]$.
The form of $M_k$ for $k\leq m$ is given in Eq.~(\ref{mkform}). Note that Eq.~(\ref{mkform}) also implies that there are at most $m$ 
(i.e. the cardinality of stabilizer group) distinct Kraus operators in the (separable) operator sum representation of $\Lambda_1$.
If there are strictly less than $m$ Kraus operators in the representation of $\Lambda_1$, 
then we can complete it to $m$ Kraus operators
by adding zeros (i.e. taking some of the $p_k$s in Eq.(\ref{mkform}) to be zero). With these notations, the map $\Lambda_1$
converts $\psi_1$ to $\psi_2$ with probability
$$
\text{Tr}\left[\Lambda_1(\psi_1)\right]=\sum_{k=1}^{m}p_k\;.
$$
Now, the completeness relation for the tace-preserving CP map $\Lambda$, can be written as,
\begin{align*}
I & =\sum_{k}M_{k}^{\dag}M_{k}=\sum_{k=1}^{m}M_{k}^{\dag}M_{k}+\sum_{k>m}M_{k}^{\dag}M_{k}\\
&=\|g_1\psi\|^2
\sum_{k=1}^{m}p_k\left(g_{1}^{\dag}\right)^{-1}U_{k}^{\dag}\rho_2(\phi)U_{k}g_{1}^{-1}+\sum_{k>m}M_{k}^{\dag}M_{k}\;,
\end{align*}
where we have used Eq.~(\ref{mkform}). Hence, the probabilities $\{p_k\}$ for $\Lambda_1$ are chosen such that
$$
\Delta_{\psi_1\to\psi_2}\left(\{p_k\}\right)=\frac{1}{\|g_1\psi\|^2}g_{1}^{\dag}\left(\sum_{k>m}M_{k}^{\dag}M_{k}\right)g_1\;.
$$ 
That is, $\Delta_{\psi_1\to\psi_2}\left(\{p_k\}\right)$ is a positive semi-definite separable matrix. This completes the proof.
\end{proof}

The requirement in theorem~\ref{delta} that $\Delta_{\psi_1\to\psi_2}\left(\{p_k\}\right)$ is a separable matrix implies that
$\mathcal{G}(\Delta_{\psi_1\to\psi_2}\left(\{p_k\}\right))$ is also a separable matrix. This leads to the following necessary condition:
\begin{corollary}\label{NC}
Notations as in theorem~\ref{delta}. If $P_m\equiv P_{\max}(\psi\to\phi)$ then the matrix
\begin{equation}\label{nc}
\mathcal{G}(\rho_1)-P_m\mathcal{G}\left(\rho_2\right)
\end{equation}
is separable.
\end{corollary}
As we will see shortly, this leeds to an upper bound on $P_m$.

\subsubsection{A simple formula for states with unitary stabilizer}

If the state $\psi_2$ in theorem~\ref{delta} is the critical state $\psi$, then $\rho_2=I$ and 
the expression $\Delta_{\psi_1\to\psi_2}\left(\{p_k\}\right)$ in theorem~\ref{delta} is given simply by
$$
\Delta_{\psi_1\to\psi_2}\left(\{p_k\}\right)=\rho_1-\sum_{k=1}^{m}p_k \;I\;.
$$
Now, from the similar arguments as in lemma~\ref{flem} it follows that $\Delta_{\psi_1\to\psi_2}\left(\{p_k\}\right)$ is separable if and only if
$\sum_{k=1}^{m}p_k$ is not greater than the smallest eigenvalue of $\rho_1$. We therefore have the following corollary:
\begin{corollary}
Let $\psi\in\mathcal{H}_n$ be a critical state and let $\psi_1\in\mathcal{O}_\psi$. Then, the maximum possible probability
to convert $\psi_1$ to $\psi$ is given by
$$
P_{\max}(\psi_1\to\psi)=\lambda_{\min}[\rho_1]\;,
$$
where $\lambda_{\min}[\rho_1]$ is the minimum eigenvalue of $\rho_1$.
\end{corollary}
Note that in this case the upper bound given in corollary~\ref{vidupper} is saturated by this value. 
That is, the probability to convert a state to the maximally entangled state of the SLOCC orbit (i.e. the critical state)
is given by the same formula that it is given in the bipartite case. 

The expression in theorem~\ref{delta} can also be simplified if 
$\psi_1$ is critical or if the stabilizer group of $\psi_1$ is unitary; i.e. if $\psi_1$ is also
a natural representative of the orbit $\mathcal{O}_{\psi}$.
Since most states are generic, their stabilizer group is finite, and if they are also critical then from proposition~\ref{nolan1} it follows that their stabilizer group is unitary. For such states the following corollary applies.

\begin{corollary}
Add to the assumptions in theorem~\ref{delta} the assumption that $Stab(\psi_1)$ is unitary.
Denote $\sigma_p\equiv \rho_1-p\mathcal{G}\left(\rho_2\right)$, where $\mathcal{G}(\cdot)$ is the $Stab(\psi)$-twirling operation (see Eq.~(\ref{gtwirling})). Then,
$$
P_{\max}(\psi_1\to\psi_2)=\max\left\{p\;\Big|\;\sigma_p\text{ separable}\right\}\;.
$$
\end{corollary}
\begin{remark} Note that $\sigma_p=\mathcal{G}\left(\rho_1-p\rho_2\right)$ since we assume that  $Stab(\psi_1)$ is unitary.
Moreover, if $p$ is small enough (see lemma~\ref{flem}) then $\rho_1-p\rho_2$ is a positive semi-definite separable matrix, 
and therefore $\mathcal{G}\left(\rho_1-p\rho_2\right)$ is also positive semi-definite separable matrix. This observation provides 
a lower bound on $P_{\max}$, which will be discussed after the proof.\end{remark}

\begin{proof}
Based on theorem~\ref{delta}, we require that  
$\Delta_{\psi_1\to\psi_2}\left(\{p_k\}\right)$ is separable operator. However, since $\mathcal{G}(\rho_1)=\rho_1$ (see proposition~\ref{proasso}) we get that
\begin{align}
&\mathcal{G}(\Delta_{\psi_1\to\psi_2}\left(\{p_k\}\right))=\mathcal{G}(\rho_1)-p\mathcal{G}\left(\rho_2\right)\nonumber\\
&=\rho_1-p\mathcal{G}\left(\rho_2\right)=\Delta_{\psi_1\to\psi_2}\left(\{p/m\}\right))
\end{align}
is also a separable matrix. Here, $p\equiv\sum_{k=1}^{m}p_k$ and $\{p/m\}$ is a set of $m$ probabilities all equal to $p/m$.
This completes the proof.
\end{proof}

\subsubsection{Lower and Upper bounds}

The formula in theorem~\ref{delta} for the maximum probability to convert one state to another is very hard to calculate especially
if the stabilizer of $\psi_1$ is not unitary.
One of the reasons for that is that the optimization in Eq.~(\ref{opt}) needs to be taken over all sets $\{p_k\}$ such that 
$\Delta_{\psi\to\phi}\left(\{p_k\}\right)$ is separable; but to check whether a multipartite mixed state is separable or not
is an NP hard problem~\cite{Gur}. Therefore, it would be very useful to find lower and upper bounds to $P_{\max}$ in terms of much 
simpler expressions. Fortunately, we were able to find such expressions.

\subsubsection*{An Upper Bound}
Notations as in theorem~\ref{delta}. A simple upper bound to $P_m$ follows immediately from corollary~\ref{NC}.
The requirement that the matrix in Eq.~(\ref{nc}) is separable implies that the matrix is also positive semi-definite.
We therefore have the following upper-bound:
\begin{corollary}
The maximum probability to convert $\psi_1$ to $\psi_2$ is bound from above by:
$$
P_{\max}(\psi_1\to\psi_2)\leq\frac{\lambda_{\max}\left[\mathcal{G}(\rho_1)\right]}{\lambda_{\max}\left[\mathcal{G}\left(\rho_2\right)\right]}
$$
where $\lambda_{\max} [\cdot]$ is the maximum eigenvalue of the matrix in the square parenthesis.
\end{corollary}

Note that another upper bound can be found by replacing the requirement in theorem~\ref{delta} that
$\Delta_{\psi\to\phi}\left(\{p_k\}\right)$ is separable, with the weaker requirement that it is only 
positive semi-definite matrix. However, this will still lead to a non-trivial optimization problem. 

\subsubsection*{A Lower Bound}

The maximum probability to convert one state to another is always greater than zero if both states belong to the same orbit.
In the following corollary we find a positive lower bound for this probability.
\begin{corollary}
Notations as in theorem~\ref{delta}. The maximum probability to convert $\psi_1$ to $\psi_2$ is bound from below by:
$$
P_{\max}(\psi_1\to\psi_2)\geq\frac{1}{\lambda_{\max}\left[\rho_{1}^{-1}\rho_2\right]}
$$
where $\lambda_{\max} [\cdot]$ is the maximum eigenvalue of the matrix in the square parenthesis.
\end{corollary}
\begin{remark}
Since the ADM is defined up to a conjugation by an element in the stabilizer group,
we can replace $\rho_2$ in the equation above by $U_{k}^{\dag}\rho_1U_{k}^{\dag}$ for any $U_k\in Stab(\psi)$. Since 
$\lambda_{\max}\left[\rho_{1}^{-1}\rho_2\right]$ is in general \emph{not} invariant under such a conjugation, we can
improve the bound by taking an optimiztion over all $m$ elements of the stabilizer; that is,
$$
P_{\max}(\psi_1\to\psi_2)\geq\max_{k=1,...,m}\frac{1}{\lambda_{\max}\left[\rho_{1}^{-1}U_{k}^{\dag}\rho_2U_k\right]}\;.
$$
Note also that 
$$
\frac{1}{\lambda_{\max}\left[\rho_{1}^{-1}\rho_2\right]}\geq\frac{\lambda_{\max}\left[\rho_1\right]}{\lambda_{\max}\left[\rho_2\right]}\;.
$$
with equality if $[\rho_1,\rho_2]=0$. Hence, the maximum eigenvalue of the associate density matrix can only increase under SEP.
This implies, for example, that $-\log\lambda_{\max}\left[\rho_1\right]$ is an entanglement monotone. In the bipartite case it is called
entanglement of teleportation~\cite{Gou04}. 
\end{remark}

\subsection{Infinite Stabilizer Group}\label{sec:conti}

In this section we generalize some of the results obtained in the previous sections for the case where the stabilizer group is not finite. Although, for most states the stabilizer group is finite, for many important states, like the GHZ type states, W type states, or Cluster states, the stabilizer group has non-zero dimension. Furthermore, as we will in the examples of the next section, the larger the stabilizer group is, the more conversions between multipartite states can be achieved by local separable operations. In this context, the fact, for example, that the maximally bipartite entangled state can be converted by LOCC to any other bipartite state is a consequence of a very large stabilizer group (see Eq.~(\ref{bistab})).

When the stabilizer group of a state $\psi$ is infinite, it is impossible in general to find a state in the orbit of $\psi$ for which the stabilizer group is a unitary group. A simple example for that can be found in the stabilizer of the three qubits GHZ state discussed in the next section. Nevertheless, theorem~\ref{two} in section~\ref{sec:niel} can be easily generalized to infinite groups, but the conditions are more complicated. 
 
 \begin{theorem}\label{cont}
Let $|\psi_1\rangle=g_1|\psi\rangle/\|g_1\psi\|$ and $|\psi_2\rangle=g_2|\psi\rangle/\|g_2\psi\|$
be two entangled states in the same orbit $\mathcal{O}_{\psi}$.  
Then, $|\psi_1\rangle\rightarrow|\psi_2\rangle$ by SEP 
if and only if there exists $m\in\mathbb{N}$ and a set of probabilities $p_0,p_1,...,p_m$ (i.e. $p_k\geq 0$ and $\sum_{k=1}^{m}p_k=1$)  such that
\begin{equation}\label{basiccont}
\sum_{k=1}^{m}p_kS_{k}^{\dag}\rho_2S_{k}=\rho_1\;,
\end{equation}
for some $S_k\in Stab(\psi)$, where $\rho_i\equiv\rho_{\psi}(\psi_i)$ for $i=1,2$.
\end{theorem}
 
The proof of the theorem above follows the exact same lines as in the proof of theorem~\ref{two}.
Similarly, theorem~\ref{basic2} can also be generalized to non-compact groups in this way.
However, since the set of $m$ elements, $\{S_k\}$, is only a finite subset of the infinite group $Stab(\psi)$,
it can be very difficult to find such a finite subset of $Stab(\psi)$ that satisfies Eq.(\ref{basiccont}). 
Moreover, note that according to proposition~\ref{noncom} the $Stab(\psi)$ is also non-compact
which makes it impossible to generalize the stabilizer twirling operation (see Eq.(\ref{gtwirling})) to this case. 
Nevertheless, we will see in the 3 qubits example of the next section that the twirling operation in Eq.(\ref{gtwirling2}) 
can be a very useful tool in determining local conversions of multipartite states with non-compact stabilizers.
The idea to define the group $T$ as the intersection of $Stab(\psi)$ with the group of local unitary matrices is
useful in general (not only for the 3 qubits example below), but we leave the general analysis for future work.

\section{Examples}\label{sec:ex}

In this section we give examples demonstrating the applications of the main results to specific cases.
We start by showing how our techniques produce all the known results about the necessary and sufficient conditions for 
local interconversion among pure bipartite states. 
Then, we focus on Hilbert spaces of three and four qubits. 
We start with the four qubits case because the stabilizer
of generic states in four qubits is finite, while the stabilizer of the three qubits GHZ state turns out to be a two dimensional non-compact group, and therefore the analysis is a bit more complicated. 

\subsection{The Bipartite Case}\label{bip}
In this subsection we consider the bipartite example with a Hilbert space 
$\mathcal{H}_2=\mathbb{C}^{d}\otimes\mathbb{C}^{d}$. Up to local unitaries, 
the only critical state is the maximally entangled state
$$
|\psi\rangle=\sum_{i=0}^{d-1}|i\rangle |i\rangle\;,
$$
where $\{|i\rangle\}$ is an orthonormal basis in $\mathbb{C}^d$. Clearly, in this case the orbit $\mathcal{O}_{\psi}$ is dense and
consists of all states in $\mathbb{C}^{d}\otimes\mathbb{C}^{d}$ with Schmidt number $d$.
In order to apply the theorems from the previous sections, we need to calculate first the stabilizer group. 

\begin{lemma}
The stabilizer group of $\psi$ is given by 
\begin{equation*}
Stab(\psi)=\left\{S^{-1}\otimes S^{T} \Big| \;S\in GL(\mathbb{C}^{d})\right\}\;,
\end{equation*}
where $S^{T}$ is the transpose matrix of $S$. 
\end{lemma}
\begin{proof}
It is straightforward to check that $S^{-1}\otimes S^{T}|\psi\rangle=|\psi\rangle$. 
Now, if $A\otimes B\in Stab(\psi)$ then
$$
|\psi\rangle=A\otimes B|\psi\rangle=\left(A\otimes B\right)\left(A^{-1}\otimes A^{T}\right)|\psi\rangle
=I\otimes BA^{T}|\psi\rangle\;.
$$
Thus, $BA^{T}=I$ which implies that $A\otimes B$ is of the form given in the lemma.
\end{proof}

Clearly, the lemma above implies that $Stab(\psi)$ is not compact. 
However, in~\cite{Lo97} it was shown that when considering only pure bipartite states, 
any LOCC transformation can be achieved by the following protocol:
Alice performs a quantum measurement and send the result to Bob. Based on this result, Bob performs a unitary matrix.
Hence, all LOCC protocols can be achieved by the following completely positive map. For all 
$|\psi\rangle\in\mathbb{C}^{d}\otimes\mathbb{C}^{d}$, 
$$
\Lambda\left(|\phi\rangle\langle\phi|\right)=\sum_{k}M_k|\phi\rangle\langle\phi|M_{k}^{\dag}\;\;,\;M_k=A_k\otimes B_k\;,
$$
where $B_k$ is a unitary matrix. 
Thus, without loss of generality, for the bipartite case we can replace $Stab(\psi)$ in theorem~\ref{cont}
with the subgroup $T\subset Stab(\psi)$ given by
\begin{equation*}
T\equiv\left\{V^{\dag}\otimes V^{T} \Big| \;V\in U(d)\right\}\;.
\end{equation*}
Note that $T$ is a compact unitary group.

Now, let $|\psi_1\rangle$ and $|\psi_2\rangle$ be as in theorem~\ref{two} applied to the bipartite case, and denote by $\rho_1$ and $\rho_2$ their ADMs, respectively.
Towards the end of Sec.~\ref{sec:asso} we have seen that the ADMs $\rho_1$ and $\rho_2$ 
can be written as
$$
\rho_i=\rho^{(i)}_{r}\otimes I\;\;,\;\text{for }i=1,2\;\;,
$$
where $\rho^{(1)}_{r}$ and $\rho^{(2)}_{r}$ are the the reduced density matrices of $|\psi_1\rangle$ and $|\psi_2\rangle$, respectively. Thus, theorem~\ref{cont} implies that the map $|\psi_1\rangle\to|\psi_2\rangle$ can be implemented by LOCC
if and only if their exists probabilities $\{p_k\}$ and elements $S_k\in T$ such that
$$
\sum_{k=1}^{m}p_kS_{k}^{\dag}\rho_2S_{k}=\rho_1\;.
$$
But this condition is equivalent to
$$
\sum_{k=1}^{m}p_kV_{k}^{\dag}\rho^{(2)}_{r}V_{k}=\rho^{(1)}_{r}
$$
where $\{V_k\}$ is a set of $m$ unitary matrices. This is the exact same condition obtained in~\cite{Nie99}
for pure bipartite LOCC transformations and is equivalent to the majorization condition. In the same way, we can show that when applying theorem~\ref{basic2} to the bipartite case, we can obtain the same results found in~\cite{Ple99} for non-deterministic LOCC transformations.

\subsection{Four qubits}

In this subsection we consider the Hilbert space of four qubits 
$\mathcal{H}_{4}=\mathbb{C}^{2}\otimes\mathbb{C}^{2}\otimes\mathbb{C}^{2}\otimes\mathbb{C}^{2}$.
The SLOCC group is given by $G=SL(2,\mathbb{C})\otimes SL(2,\mathbb{C})\otimes SL(2,\mathbb{C})\otimes SL(2,\mathbb{C})$.
The analysis of four qubits SLOCC classes can be found in~\cite{Ver02,Ver03,GW10,W} (and references therein), and we refer the reader to these references for more details about some of the statements we make here.

In four qubits the critical set can be characterize elegantly by the following 4 dimensional subspace: 
$$
Crit(\mathcal{H}_4)=K\mathcal{A},
$$ 
where $K\equiv SU(2)\otimes SU(2)\otimes SU(2)\otimes SU(2)$ is the set of special local unitary matrices,
$$
\mathcal{A}\equiv\Big\{z_0u_0+z_1u_1+z_2u_2+z_3u_3\Big|\;z_0,z_1,z_2,z_3\in\mathbb{C}\Big\}\;,
$$
with 
\begin{align}
& u_0\equiv|\phi^{+}\rangle|\phi^{+}\rangle\;\;,\;\;u_1\equiv|\phi^{-}\rangle|\phi^{-}\rangle\nonumber\\
& u_2\equiv|\psi^{+}\rangle|\psi^{+}\rangle\;,\;\;
u_3\equiv|\psi^{-}\rangle|\psi^{-}\rangle\nonumber
\end{align}
and $|\phi^{\pm}\rangle=(|00\rangle\pm|11\rangle)/\sqrt{2}$ and $|\psi^{\pm}\rangle=(|01\rangle\pm|10\rangle)/\sqrt{2}$.
In~\cite{Ver02} the class $\mathcal{A}$ was denoted by $G_{abcd}$.

Note that not all the states in $\mathcal{A}$ are generic. For example, the state $u_0$ is not generic.
It is possible to show~\cite{W} that a four-qubit state $\phi$ is generic if and only if $(a)$ the dimension of the (non-normalized) orbit $G\phi$ is maximal (i.e. $12$), and $(b)$ the orbit $G\phi$ is closed in $\mathcal{H}_4$.
 
The set of all generic states in $Crit(\mathcal{H}_4)$ is given by $K\mathcal{D}$, where
$$
\mathcal{D}\equiv\Big\{z_0u_0+z_1u_1+z_2u_2+z_3u_3\Big|z_{i}^{2}\neq z_{j}^{2}\;\;\text{for}\;\;i\neq j\Big\}
$$
In the following theorem we prove that for all these states the intersection of stabilizer group with $G$ is the Klein group, consisting of four elements.
\begin{proposition}
Let $\psi\in\mathcal{D}$ and denote by $St(\psi)$ the intersection of $Stab(\psi)$ with $G$ (see corollary~\ref{stst}). Then, 
$$
St(\psi)=\{I,\bar{X},\bar{Y},\bar{Z}\}\;,
$$
where $X$, $Y$ and $Z$ are the three $2\times 2$ Pauli matrices and $\bar{C}\equiv C\otimes C\otimes C\otimes C$
for $C=X,Y,Z$.
\end{proposition}
\begin{proof}
If $\phi\in A$ and $\phi=\sum_{j=0}^{4}z_{j}u_{j}$ and $z_{i}^{2}\neq
z_{j}^{2}$ for $i\neq j$ then we say that $\phi$ is generic. If $g\in G$ and
$gA=A$ then if $\beta=\sum x_{i}u_{i}$ then $g\beta=\sum\varepsilon
_{i}x_{\sigma i}u_{i}$ with $\sigma$ a permutation, $\varepsilon_{i}=\pm1$ and
$\varepsilon_{0}\varepsilon_{1}\varepsilon_{2}\varepsilon_{3}=1$. We will show
that the stabilizer in $G$ of a generic element in $A$ has the property that it
leaves $A$ invariant. We therefore see that if $g\phi=\phi$ with $\phi$
generic then $g_{|A}=I$. We now give more details (see~\cite{W}).

To complete the discussion we reformulate the action of $G$ on $\mathcal{H}
_{4}$. We first consider $SL(2,\mathbb{C})\otimes SL(2,\mathbb{C})$ acting on $\mathcal{H}_{2}$. There is a non-degenerate symmetric form,
$(...,...)$ on $\mathcal{H}_{2}$ that is invariant under $SL(2,\mathbb{C})\otimes SL(2,\mathbb{C})$ (take the tensor product of the symplectic form on $\mathbb{C}^{2}$). Relative to an orthonormal basis of $\mathcal{H}_{2}$ the image of
$SL(2,\mathbb{C})\otimes SL(2,\mathbb{C})$ is $SO(4,\mathbb{C})$. Thus the action of $G$ on $\mathcal{H}_{4}=\mathcal{H}_{2}\otimes
\mathcal{H}_{2}$ is equivalent to the action of $SO(4,\mathbb{C})\otimes SO(4,\mathbb{C})$ on $\mathbb{C}^{4}\otimes\mathbb{C}^{4}$. We can use the form $(...,...)$ to identify $\mathbb{C}^{4}\otimes\mathbb{C}^{4}$ with $M_{4}(\mathbb{C})$ with \ $SO(4,\mathbb{C})\times SO(4,\mathbb{C})$ acting by $(g,h)X=gXh^{T}=gXh^{-1}$. \ One can check that with these
identifications we can take $A$ to be the diagonal elements. If we consider the
group $SO(8,\mathbb{C})$ then we can consider $G$ to be the subgroup
\[
\left\{  \left[
\begin{array}
[c]{cc}
g & 0\\
0 & h
\end{array}
\right]  |g,h\in SO(4,\mathbb{C})\right\}  .
\]
The space $\mathcal{H}_{4}$ can be imbedded in $Lie(SO(4,\mathbb{C}))$ as
\[
\left\{  \left[
\begin{array}
[c]{cc}
0 & X\\
-X^{T} & 0
\end{array}
\right]  |X\in M_{4}(\mathbb{C})\right\}  .
\]
The space $A$ now corresponds to a Cartan subalgebra of 
$Lie(SO(8,\mathbb{C}))$ and the assertion above is implied by the fact that the Weyl group of
$SO(8,\mathbb{C})$ is given as above. This implies that the stabilizer of a regular element in
$A$ fixes every element in $A$. In particular if $(g,h)$ is in the stabilizer
then
\[
gIh^{T}=I
\]
so $g=h$. Also if for all $X$ we have $gXg^{-1}=X$ then $g$ is diagonal. Since
$g$ is orthogonal $gg^{T}=I$ so $g^{2}=I$. Hence $g$ is diagonal with entries
$\pm1$. Now unwind the identifications.
\end{proof}

In 4-qubits there are 4 generating $SL$ invariant polynomials~\cite{GW10}. One of these generating polynomials
is of order 2 and is given by:
$$
f(\psi)\equiv\langle\psi^{*}|
\sigma_y\otimes\sigma_y\otimes\sigma_y\otimes\sigma_y|\psi\rangle\;.
$$ 
The absolute value of this SL-invariant polynomial is called the 4-tangle. If $\psi\in Crit(\mathcal{H}_4)$
then
$$
f(\psi)=z_{0}^{2}+z_{1}^{2}+z_{2}^{2}+z_{3}^{2}\;.
$$
Now, from proposition~\ref{gk} (see also the remark after the proof of proposition~\ref{gk}) it follows that if 
$f(\psi)\neq 0$ then
$$
St(\psi)=Stab(\psi)\;.
$$
That is, for all the generic states in $\mathcal{D}$ with $f(\psi)\neq 0$, the stabilizer group is the 4 element
Klein group $Stab(\psi)=\{I,\bar{X},\bar{Y},\bar{Z}\}$. Hence, the results in~\ref{sec:niel} implies that a state 
$\psi\in\mathcal{D}$ with $f(\psi)\neq 0$, can be converted to $|\phi\rangle\equiv g|\psi\rangle/\|g|\psi\rangle$ by SEP,
if and only if
$$
\mathcal{G}(\sigma)=I\;\;,\;\sigma\equiv\frac{g^{\dag}g}{\|g|\psi\rangle\|^2}\;,
$$
where $\mathcal{G}$ is the $Stab(\psi)$-twirling operation defined in Def.~\ref{gtwirling}. Moreover, 
it is possible to characterize completely \emph{all} the ADMs $\sigma$ that satisfy the equation above.
This can be done by applying the factorization in Eq.(\ref{factorization}) to the Klein group. 

For states $\psi\in\mathcal{D}$ with $f(\psi)=0$ one can use corollary~\ref{stst} to find out the stabilizer group.
For example, the state (see~\cite{GW10} for important properties of this state)
$$
|L\rangle\equiv\frac{1}{\sqrt{3}}\left(u_0+\omega u_1+\bar{\omega}u_2\right)\;,
$$
is in $D$ and satisfies $f(|L\rangle)=0$. In fact it is possible to show that it vanish on all $SL$-invariant polynomials
of degree less than 6, and that there exists an $SL$-invariant polynomial of degree 6 on which it does not vanish~\cite{GW10}.
Hence, from corollary~\ref{stst} it follows that $Stab(|L\rangle)/St(|L\rangle)$ is a cyclic group with three elements.
The generator of the group is given by $g=\omega p_3$, where $p_3\in K$ is the local unitary permutation that takes $u_0\to u_2$,
$u_1\to u_0$, $u_2\to u_1$ and $u_3\to u_3$. Hence, now that we found $Stab(|L\rangle)$, we can apply all the theorems in the previous section to determine the states to which $|L\rangle$ can be converted by SEP. 
 
\subsection{Three Qubits}

In this subsection we apply the results from the previous section to the case of three qubits.
In particular, we will consider local transformations between two states in the GHZ SLOCC class~\cite{Vid00}.
Since the stabilizer group of the GHZ state is non-compact, we will only be able to solve the problem partially.

We denote the GHZ state by $\phi=\frac{1}{\sqrt{2}}(\left\vert 000\right\rangle
+\left\vert 111\right\rangle )$. A straightforward calculation shows that the stabilizer of $\phi$ in $G$ is the
group
\[
St(\phi)=\left\{  t_{1}\otimes t_{2}\otimes t_{3}\;\Big|\;t_{j}=\left[
\begin{array}
[c]{cc}
s_{j} & 0\\
0 & s_{j}^{-1}
\end{array}
\right]  ,s_{1}s_{2}s_{3}=1\right\}  .
\]
Note that if $t=t_{1}\otimes t_{2}\otimes t_{3}$ as above then we have
\begin{align*}
t\left\vert 000\right\rangle  &  =\left\vert 000\right\rangle ,t\left\vert
111\right\rangle =\left\vert 111\right\rangle \\
t\left\vert 001\right\rangle  &  =s_{3}^{-2}\left\vert 001\right\rangle
,t\left\vert 010\right\rangle =s_{2}^{-2}\left\vert 010\right\rangle
,t\left\vert 100\right\rangle =s_{1}^{-2}\left\vert 100\right\rangle ,\\
t\left\vert 011\right\rangle  &  =s_{1}^{2}\left\vert 011\right\rangle
,t\left\vert 101\right\rangle =s_{2}^{2}\left\vert 101\right\rangle
,t\left\vert 110\right\rangle =s_{3}^{2}\left\vert 110\right\rangle .
\end{align*}

Now, the 3-tangle of the GHZ state is non-zero. Since the tangle is the absolute value of an homogeneous $SL$-invariant polynomial of degree 4, from corollary~\ref{stst} it follows that $Stab(\phi)/St(\phi)$ is a cyclic group with two elements.
The two elements are
$$
Stab(\phi)/St(\phi)=\left\{I\;,\;X\otimes X\otimes X\right\}\;,
$$
where $X$ is the x-Pauli matrix. $Stab(\phi)$ is therefore the union of $St(\phi)$ and $hSt(\phi)$, where 
$h\equiv X\otimes X\otimes X$.

Set $T$ equal to the intersection of $Stab(\phi)$ with the group of local unitary
transformations, and $T_0$ the intersection of $St(\phi)$ with the group of local unitary
transformations . The $T$-twirling operation is therefore
$$
\mathcal{T}(\sigma)=\frac{1}{2}\mathcal{T}_0(\sigma)+\frac{1}{2}h\mathcal{T}_0(\sigma)h
$$
where
$$
\mathcal{T}_0(\sigma)\equiv\int dt\;t^{\dag}\sigma t
$$
with $dt=\frac{1}{4\pi^{2}}d\theta_{1}d\theta_{2}$ for $s_{j}=e^{i\theta_{j}
},\theta_{1}+\theta_{2}+\theta_{3}=0$. 

Now, from theorem~\ref{cont} it follows that if $\psi_1$ and $\psi_2$ are two states in the GHZ class, then it is possible to convert $\psi_1$ to $\psi_2$ by SEP if and only if
there exists probabilities $\{p_k\}$ and elements $g_k\in Stab(\phi)$ such that
\[
\sum_{k} p_{k}g_{k}^{\dag}\rho_2 g_{k}=\rho_1\;,
\]
where $\rho_1$ and $\rho_2$ are the ADMs of $\psi_1$ and $\psi_2$, respectively.
In order to get a simpler condition we apply the $T_0$-twirling operation on both sides of this equation.
This gives, 
\[
\sum_k p_{k}\mathcal{T}_0(g_{k}^{\dag}\rho_2 g_{k})=\mathcal{T}_0\left(\rho_1\right).
\]
Since $St(\phi)$ is commutative and $hSt(\phi)=St(\phi)h$ we get
\[
\sum_k p_{k}g_{k}^{\dag}\mathcal{T}_0(\rho_2)g_{k}=\mathcal{T}_0\left(\rho_1\right).
\]
Now, without loss of generality assume $g_k\in St(\phi)$ for $k\leq a$, and $g_k=hf_k$ with $f_k\in St(\phi)$
for $k>a$. Then, 
$$
\sum_k p_{k}g_{k}^{\dag}\mathcal{T}_0(\rho_2)g_{k}
=\sum_{k\leq a}p_{k}g_{k}^{\dag}g_{k}\mathcal{T}_0(\rho_2)+\sum_{k>a} p_{k}f_{k}^{\dag}f_k\mathcal{T}_0(h\rho_2h)\;,
$$
where we have used the fact that $\mathcal{T}_0(\rho_2)$ commutes with every element of $T_0$,  and since $St(\phi)$ is the 
complexification of $T_0$ we see that $\mathcal{T}_0(\rho_2)$ commutes with every element of $St(\psi)$. 
This
implies that
\[
\sum_{k\leq a}p_{k}g_{k}^{\dag}g_{k}\mathcal{T}_0(\rho_2)+\sum_{k>a} p_{k}f_{k}^{\dag}f_k\mathcal{T}_0(h\rho_2h)
=\mathcal{T}_0\left(\rho_1\right).
\]
In the case where $\mathcal{T}_0\left(\rho_1\right)=\rho_1$, the above condition is both necessary and sufficient.
In particular, if we take $\psi_1$ to be the GHZ state $\phi$, then $\rho_1=I$ and we get
$$
\sum_{k\leq a}p_{k}g_{k}^{\dag}g_{k}\mathcal{T}_0(\rho_2)+\sum_{k>a} p_{k}f_{k}^{\dag}f_k\mathcal{T}_0(h\rho_2h)=I\;.
$$
The above condition provides necessary and sufficient condition for local separable conversion of the GHZ state $\psi_1=\phi$ to some other state $\psi_2$. In particular, the condition implies that the GHZ state can be converted by SEP to any state $\psi_2$ with 
$\mathcal{T}(\rho_2)=I$ since we can take in this case $g_k$ and $f_k$ to be unitaries.
However, it is left open if there are other states that can be obtained by SEP from the GHZ state.

\section{Conclusions}\label{sec:con}

Quantum information can be viewed as a theory of interconversions among different resources.
Multipartite entangled states, like cluster states, are resources for quantum information and quantum computing.
In this paper we considered transformations that can be implemented on pure multipartite entangled states 
by local separable operations. Quite remarkably, we were able to  generalize both Nielsen majorization theorem~\cite{Nie99} 
and Jonathan and Plenio theorem~\cite{Ple99} for transformations involving multipartite pure states. This generalization was 
expressed in terms of the the stabilizer group and the associate density matrices (ADMs) of the states involved.

For the bipartite case with a Hilbert space $\mathbb{C}^{d}\otimes\mathbb{C}^{d}$, 
both the majorization theorem in~\cite{Nie99} and the conditions in~\cite{Ple99} can be expressed in terms
of Vidal's entanglement monotones~\cite{Vid99} defined by
$$
E_k(\psi)=\sum_{j=k}^{d}\lambda_k\;\;k=1,2,...,d
$$
where $\lambda_k$ are the ordered (from the biggest to the smallest) eigenvalues of the reduced density matrix, and $\psi$
is a bipartite state. Here we were able to extend this definition to the multipartite case by taking the $\lambda_k$s to be the ordered
eigenvalues of the ADM of $\psi$. However, unlike the bipartite case, for local transformations involving multipartite states, the conditions given in terms of these functions are only necessary but not sufficient. It is therefore impossible to express the necessary \emph{and} sufficient conditions only in terms of functions of the eigenvalues of the ADMs involved in the transformation. Instead,
we express the conditions for local transformations in terms of matrix equalities given in theorems~\ref{two} and~\ref{basic2}.

The necessary and sufficient conditions gets extremely elegant for transformations that involve the maximally entangled state of the SLOCC orbit. That is, every SLOCC orbit with a finite stabilizer group has a maximally entangled state (this is usually a critical state, i.e. a normal form). Its stabilizer group consisting of unitary matrices, and the necessary and sufficient conditions to convert this state
to another state in its SLOCC orbit, are given in Eq.(\ref{abc}) and expressed elegantly in terms of the stabilizer twirling operation (see definition~\ref{gtwirling}). This simple expression made it possible to characterize \emph{all} the states to which a maximally entangled state can be converted by SEP.

Both theorem~\ref{two} and~\ref{basic2} assume that the stabilizer group of the state representing the SLOCC orbit is finite.
We also argued that for a dense set of states in $\mathcal{H}_n$ with $n>3$ the stabilizer is finite. Therefore, these theorems hold true for most states in $\mathcal{H}_n$. Moreover, from continuity arguments, it follows that the necessary conditions in Eq.(\ref{ensvid}) given in terms of the continuous functions $E_k(\psi)$ hold true for \emph{all} states in $\mathcal{H}_n$.
That is, the necessary conditions that generalize the majorization theorems in~\cite{Nie99} and~\cite{Ple99} hold true for all the states in the Hilbert space. In particular, this implies that $E_{k}$ are entanglement monotones.

When local deterministic transformations are not possible, we found an expression for the maximum probability to convert one multipartite state to another by SEP. This expression can be simplified dramatically if one of the two states involved is critical.
However, since in general the expression is relatively complicated, 
we found simple lower and upper bounds for the probability of local conversion.

There are many important SLOCC orbits, like the GHZ class, that have a non-finite (in fact, non-compact) stabilizer group. 
For such orbits we also found necessary and sufficient conditions for a local conversion among states in the same orbit, but these
conditions are more complicated since the stabilizer group is not unitary. We are willing to conjecture that for such orbits 
theorems~\cite{Nie99} and~\cite{Ple99} hold true if one replace the non-compact stabilizer group, with the subgroup obtained by the intersection of the stabilizer with the local unitary group. 

There are many other interesting related problems that where not discussed here.
One such problem concerns with the asymptotic rates to interconvert one multipartite state to another.
We hope the work presented here will be useful for such a research direction.

\emph{Acknowledgments:---}
The authors acknowledge support from iCore for Nolan Wallach visit to IQIS in Calgary.
G.G. research is supported by NSERC.
The research of N.W. was partially supported by the NSF.

\end{document}